\documentclass[a4paper,twocolumn,preprintnumbers,amsmath,amssymb,pra]{revtex4}
\usepackage[top=2cm, bottom=2cm, left=1.5cm, right=1.5cm]{geometry}

\usepackage{graphicx}
\usepackage[font={small,it}]{caption}
\usepackage{subcaption}
\usepackage{epstopdf}

\usepackage{bm}
\usepackage{bbm}

\usepackage{amsmath}
\usepackage{amsthm}
\usepackage{amssymb}

\usepackage{lscape}
\usepackage{tabularx}
\usepackage{gensymb}

\usepackage{color}
\usepackage[dvipsnames]{xcolor}
\usepackage[normalem]{ulem}

\newcommand{\bra}[1]{\langle #1|}
\newcommand{\ket}[1]{|#1\rangle}
\newcommand{\braket}[2]{\langle #1|#2\rangle}
\newcommand{\ketbra}[2]{|#1\rangle\!\langle#2|}
\newcommand{\expt}[1]{\langle#1\rangle}
\newcommand{\id}{\mathbbm{1}}

\newtheorem*{thmnonumb}{Theorem}

\newtheorem{defin}{Definition}

\let\oldmarginpar\marginpar
\renewcommand\marginpar[1]{\-\oldmarginpar[\raggedleft\marginparsize #1]%
{\raggedright\marginparsize #1}}

\def \spek [#1]{\includegraphics[width=0.5cm]{./images/spekkens/spek#1.eps}}

\newcommand{\inlineheading}[1]{\textbf{\textit{#1---}}}
\newcommand{\inlinesubheading}[1]{\textit{#1---}}

\begin{document} 

\setlength{\tabcolsep}{1ex}

\title{A general framework for phase and interference}

\author{Andrew~J.~P.~Garner$^{1}\footnote{a.garner1@physics.ox.ac.uk}$, Oscar~C.~O.~Dahlsten$^{12\footnote{oscar.dahlsten@physics.ox.ac.uk}}$, Yoshifumi~Nakata$^{3}\footnote{Now at: Institut f\"ur Theoretische Physik, Leibniz Universit\"at Hannover, Appelstr. 2, 30167 Hannover, Germany}$, Mio~Murao$^{34}$ and Vlatko~Vedral$^{12}$}
\affiliation{
 $^{1}$\mbox{Atomic\ and\ Laser Physics, Clarendon\ Laboratory,} \\
 \mbox{University\ of\ Oxford, Parks\ Road, Oxford, OX13PU, United\ Kingdom}\\
 $^{2}$\mbox{Center\ for\ Quantum\ Technologies, National\ University\ of\ Singapore, Republic\ of\ Singapore} \\
 $^{3}$\mbox{Department of Physics, Graduate\ School\ of\ Science, University\ of\ Tokyo, Tokyo\ 113-0033, Japan}\\
 $^{4}$\mbox{Institute\ for\ Nano\ Quantum\ Information\ Electronics, University\ of\ Tokyo, Tokyo\ 153-8505, Japan}
}

\date{ \today}


\begin{abstract} 
Phase plays a crucial role in many quantum effects including interference. Phase is normally defined in terms of complex  numbers that appear when representing quantum states as complex vectors. Here we give an operational definition whereby phase is instead defined in terms of measurement statistics. Our definition is phrased in terms of the operational framework known as generalised probabilistic theories, or the convex framework. The definition makes it possible to ask whether other theories in this framework can also have phase. We apply our definition to investigate phase and interference in several example theories: classical probability theory, a version of Spekken's toy model, quantum theory, and box-world. We find that phase is ubiquitous; any non-classical theory can be said to have non-trivial phase dynamics. 
\end{abstract}

\maketitle 

\section{Introduction}
In quantum theory different {\em phases} $\phi_i$ can be associated with different branches of a superposition by polar-decomposing the complex amplitudes of each branch: $\ket{\psi}=\sum_j r_j e^{i\phi_j}\ket{j}$. Whilst they are not observable in the basis $\{\ket{j}\}$ in question they can have significant consequences for measurements in other bases. This plays a fundamental role in many of the most strikingly non-classical behaviours of quantum theory, including interference. Phases moreover play a crucial role in decoherence~\cite{Schlosshauer07} and quantum thermalisation~\cite{LindenPSW09,NakataTM12}. There is evidence that they lie at the heart of the apparent exponential speed-up of quantum computers, in that many key quantum algorithms can be phrased as performing phase-estimation~\cite{CleveEMM98}. 
It has also been argued that instantaneous quantum polynomial-time (IQP) circuits, 
which change only the phases of the separable input state with respect to the computational basis,
are likely to have stronger computational power than classical computers~\cite{ShepherdB09, BremnerJS10}.
Accordingly understanding the phenomenon of phase should be a key aim of research efforts in quantum foundations. 

The definition of phase in terms of the exponent above is not operational in its nature, as it is not defined in terms of measurement statistics but in terms of the theoretical model. This means that it is a priori not well-defined to talk about phase in experiments involving systems not governed by quantum theory. Yet there is currently great interest in parts of the quantum information/quantum foundations community in investigating theories in a wider framework than quantum theory. A key motivation for investigating such theories is to understand quantum theory better by investigating what happens if some restriction from  quantum theory is relaxed. For example there is much interest in investigating whether any fundamental physical principle would be violated if non-locality were to exceed Tsirelson's bound, which quantum theory respects~\cite{Tsirelson93,Cirelson80,PopescuR94}. It is standard to use a framework called general probabilistic theories (GPTs) for these studies (the framework is also called the convex framework) see e.g.\ ~\cite{Hardy01,Barrett07}. With this paper we aim to lay the foundations for studying phase and its impact e.g.\ on computing in GPTs. 

We define phase operationally. 
As phase is inherently a relative property of two states, we focus on phase {\em transforms}, which would in the quantum case be transforms that change the $\phi_j$ but not the $r_j$.
This means in the quantum case that the phase transforms leave the statistics of measurement in the basis $\{\ket{j}\}$ invariant. 
More generally, in the GPT framework we define the phase transforms of a measurement as those which leave the measurement statistics invariant.
As in quantum theory, one can label the phases of a {\em state} by associating it with the phase transform that creates this state from some arbitrary reference state with the same values of $r_j$ (in quantum theory, this reference state is usually one where all $\phi_j=0$).

We apply this definition to the quantum case where the standard notion of phase is recovered for maximal measurements, meaning measurements which distinguish a number of states equal to the Hilbert space dimension. We also investigate phase in (i) classical probability theory, (ii) Spekkens' toy theory---a hidden variable model which emulates many quantum effects~\cite{Spekkens07}, and (iii) box-world, a theory which allows for non-locality beyond Tsirelson's bound~\cite{Barrett07}. We find that a GPT theory has non-trivial phase transforms with respect to maximal measurements if and only if it is non-classical. We discuss the connection between phase and interference. 

A generalised definition of the phase group was proposed and studied in two papers by Coecke, Duncan, Edwards and Spekkens, formulated in a different framework based on category theory, employing diagrammatical calculus~\cite{CoeckeD11,CoeckeES11}. Our definition of the phase group in the GPT framework turns out to coincide with that other definition in both quantum theory and Spekken's toy theory (as we will show later in the paper). On the other hand we find that the phase group can in some theories be non-abelian whereas Coecke et al.'s definition demands that it be abelian. This apparent contradiction may be simply resolved if these cases are not covered by the other framework (it is to our knowledge not known whether the theory containing gbits can be formulated in the other framework). It may alternatively be that the definitions are simply different in general; this question deserves further study and is likely to illuminate the relation between the two frameworks more generally.

A description of interference in the GPT framework has been presented by Ududec, Barnum and Emerson\cite{UdudecBE09,Ududec12}, who focus their attention on triple-slit experiments, building on the work of Sorkin and others on hierarchical families of interference\cite{Sorkin94,BarnettDR07}.
Sorkin shows that in quantum theory the set of output states (i.e.~interference patterns) that have passed through all three slits can be fully described by considering combinations of states that have passed through just two of the slits, but Ududec shows for general theories this is not necessarily the case.
These works taken together suggest studying phase-related effects in more general scenarios is both possible and fruitful, and so
here we propose a systematic approach giving a general operational definition of phase.
Our discussion of interference focuses on the relationship between phase and the {\em total} set of interference dynamics. 
One can perform a decomposing analysis similar to Sorkin or Ududec by making use of the framework we present here.  Within the total group of phase transformations, one must identify subgroups that can be associated with a smaller subset of slits, using reasoning such as locality arguments\cite{DahlstenGV12}.

We proceed as follows. Section~\ref{Sec:tech_intro} gives a technical introduction to relevant aspects of generalised probabilistic theories (GPTs). In section~\ref{Sec:PG} we propose a definition of reversible phase dynamics: ``phase groups'' in GPTs and apply it to Spekken's model and gbits. In section IV we extend this to include irreversible phase dynamics, proving that phase dynamics are non-trivial iff the theory is {\em non-classical} (using the standard GPT definition of non-classical). Section~\ref{Sec:Int} discusses the relation between phase and interference in GPTs.

\section{Technical Introduction} We begin by describing GPTs, with a focus on the three examples of a classical bit, a qubit and a gbit. We then describe toy theory of Spekkens, which is a priori different, and show how it can be treated in the GPT framework. 

\label{Sec:tech_intro}
\subsection{Generalized probabilistic theories} \label{Sec:GPT}
We will find it convenient for our purposes to use the representation of GPTs from Ref.~\cite{Barrett07}. 
A GPT is defined by a tuple of {\it a state space}, {\it effects} and {\it transformations}.
A {\it state} $\vec{s}$ is completely defined by a list of probability distributions such as: 
\begin{equation}
\vec{s} = 
\begin{pmatrix}
p(a_0|M_0) \\
p(a_1|M_0) \\
\vdots \\
p(a_{l_0}|M_0) \\
\hline 
p(b_0|M_1) \\
p(b_1|M_1) \\
\vdots \\
p(b_{l_1}|M_1) \\
\hline 
\vdots \\
\hline 
p(z_0|M_{K-1}) \\
p(z_1|M_{K-1}) \\
\vdots \\
p(z_{l_{K-1}}|M_{K-1}) 
\end{pmatrix},
\end{equation}
where $\forall j$, $\{p(a_0 | M_j), \cdots, p(a_{l_j} | M_j) \}$ is a probability distribution, namely, $\forall i$, $0 \leq p(a_i|M_j) \leq 1$ and $\sum_i p(a_i|M_j)=1$. Each probability distribution $\{p(a_0 | M_j), \cdots, p(a_{l_j} | M_j) \}$ is understood as follows.
For a state $\vec{s}$, if we ``measure'' $M_j$, we obtain the ``outcome'' $a_i$ with probability $p(a_i | M_j)$. 
The definition of a state implies that, by ``measuring'' $M_0, \cdots, M_{K-1}$ and obtaining the probability distribution of ``outcomes'' for each $M_i$, we can specify all properties of the state. Such a set of ``measurements'' $\{ M_0 ,\cdots, M_{K-1}\}$ is referred to as a set of {\it fiducial measurements}. In quantum theory, any informationally complete set of measurements can be fiducial measurements. 
The set of all states is called the {\it state space}.

A {\it measurement} in GPTs is defined by a set of vectors $\{ \vec{e}_i \}_i$ satisfying that
for any state $\vec{s}$, $0 \leq \vec{e}_i \cdot \vec{s} \leq 1$ and $\sum_i \vec{e}_i \cdot \vec{s} =1$, that is,
$\vec{e}_i \cdot \vec{s}$ gives the probability to obtain the outcome $i$ when the state is $\vec{s}$. Each vector $\vec{e}_i$ in a measurement is referred to as an {\it effect}. 
A fiducial measurement is a special type of measurement such that their effects are represented by vectors of which all elements are zero but one element is one such as $(1,0, 0\cdots)^\mathrm{T}$ where $\mathrm{T}$ denotes a transposition.
 
We also define a {\it maximal measurement}. Let $N$ be the maximal number of states that can be distinguished by a single measurement. A maximal measurement is a measurement that can deterministically distinguish $N$ states by performing the measurement only once.
In quantum theory, $N$ is equal to the dimension of the Hilbert space and any rank-1 projective measurements are maximal measurements.
Whether or not a set of effects yields valid probabilities is dependent on the theory. The most general set of effects for a qubit measurement (see eqns~\ref{eq:gen_effect_1} and \ref{eq:gen_effect_2} in Appendix~\ref{app:general_effect}) can lead to probabilities below 0 or above 1 if applied to some (non-quantum) states, such as $\left(1, 0 ~|~ 1, 0 ~|~ 1, 0\right)^\mathrm{T}$.

Finally, {\it reversible transformations} in GPTs are defined by any linear maps that transform a state space to itself.
Since all such automorphisms form a group, reversible transformations in GPTs are described by a group, $G$.

In the following, we present three examples expressed in the framework of GPTs.
We summarize the state space and the automorphisms of each example in Fig.~\ref{fig:statespaces}.

\begin{figure}[ht!]
    \begin{subfigure}[t]{0.5\textwidth}
        \centering
	\includegraphics{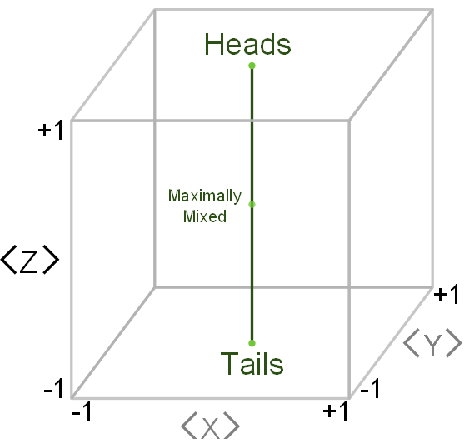} 
	\caption{{\bf Classical bit:} $C_2$ \\ This classical state space corresponds to flipping a coin. The extremal points represent the coin being known to be in the {\em heads} (Z=1) or {\em tails} (Z=-1) state. All points on the line between correspond to some uncertainty about the coin's state. 
}
	\label{fig:statespace_classical}
	~ \vspace{0.5cm}
    \end{subfigure} 
    \begin{subfigure}[t]{0.5\textwidth}
	\centering
	\includegraphics{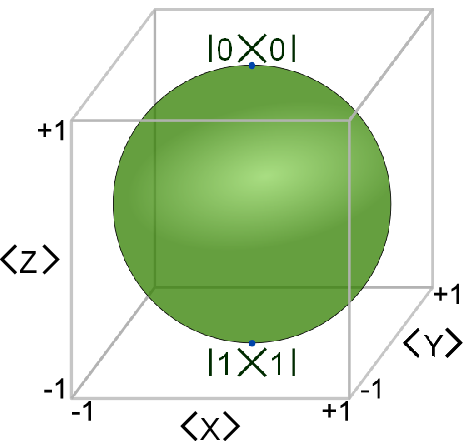} 
	\caption{{\bf Qubit:} $SO(3)$ \\ This state space is the Bloch sphere for normalised qubit states.}
	\label{fig:statespace_qubit}
	~ \vspace{0.5cm}
    \end{subfigure}
    \begin{subfigure}[t]{0.5\textwidth}
	\centering
	\includegraphics{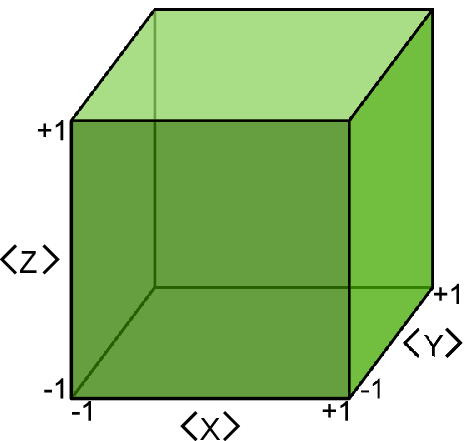} 
	\caption{{\bf $m$-in-$2$-out gbit:} $\varsigma_{2^{m-1}} \rtimes C_2$ \\ (shown: $m=3$) \\ This state space includes states not allowed in quantum theory, such as the corners, where one can simultaneously predict the outcomes for $X$,$Y$ and $Z$ measurements with certainty.}
	\label{fig:statespace_gbit}
    \end{subfigure}
\caption{The normalised state spaces of a classical bit, a qubit and an~m-in~2-out gbit}
\label{fig:statespaces}
\end{figure}

\inlineheading{Classical bits}
A classical bit (such as would be required to describe a coin flip) has a trivial structure, composed only of one measurement $M$ such that the state can be expressed by $\vec{s} = (p(0|M), p(1|M)) = (\lambda, 1-\lambda)$ where $\lambda \in [0,1]$, and the state space is a line (see Fig~\ref{fig:statespace_classical}) The states $(1,0)$ and $(0,1)$ correspond to the system being definitely in the $M=0$ (heads) or $M=1$ (tails) states, respectively; but if we are uncertain of the state of the bit (such as if we flipped a coin but hid it under our hands before checking the outcome), then it can be in a probabilistic mixture of the two. The allowed reversible transformations are to either flip the bit, or leave it unchanged: the cyclic group of degree two, $C_2$.

\inlineheading{Qubits}
A state of a normalised (density matrix of trace 1) qubit is characterized by fiducial measurements $\{X, Y, Z\}$ that correspond to measurements in
each Pauli basis~\cite{NielsenC00}. A state is represented by $\vec{s}=(p(i|W))_{\{i=0,1; W=X,Y,Z\}}$. 
Note that as $p(1|W)$ is given by $1-p(0|W)$, the statistics of a qubit measurement can be determined by just $p(0|W)$.

The commutation relations between the Pauli operators lead to the {\em uncertainty} restriction on the state space such that $\sum_{W=X,Y,Z}(p(0|W)-\frac{1}{2})^2 \leq \frac{1}{4}$. Hence, the state space is given by a $3$-dimensional ball (see Fig.~\ref{fig:statespace_qubit}).
Reversible transformations are represented by rotations of the sphere, that is, $O(3)$. 
By requiring complete positivity when the qubit is embedded into a larger system, reflections are not allowed in quantum theory\cite{Peres95} and one restricts the reversible transformations to $SO(3)$.

\inlineheading{Gbits}
{\it Gbits} are a generalization of qubits, which do not have constraints on the state space. Consider $m$ fiducial measurements $\{X_0, \cdots, X_{m-1} \}$ such that each fiducial measurement has $n$ outcomes. Such a system is called an {\it $m$-in $n$-out gbit}. A state is given by $\vec{s}=(p(i|W))_{\{i=0,1,\cdots, n-1; W=X_0,\cdots,X_{m-1} \}}$ and $p(i|W)$ takes values between $0$ and $1$ as long as it satisfies $\sum_i p(i|W) =1$.

In particular, $3$-in $2$-out gbits have been studied by analogy with qubits. Similarly to qubits,  $p(1|W)=1-p(0|W)$.
Since there is no constraint on states, the state space is a cube (see Fig.~\ref{fig:statespace_gbit}). 

The allowed reversible transformations are any maps which transform a cube to a cube (i.e. its {\em automorphism group}), which is the polyhedral group of order $6$, $\mathcal{T}_6$. Note that a polyhedral group $\mathcal{T}_6$ is isomorphic to the semi-direct product of the symmetric group of degree $4$, $\varsigma_4$, and a reflection $C_2$, namely, $\mathcal{T}_6  \cong  \varsigma_4 \rtimes C_2$. This can be easily seen from the fact that any actions of $\mathcal{T}_6$ can be expressed by permutations of four diagonal lines $\varsigma_4$ and a reflection $C_2$.
In analogy with quantum theory, we can consider a theory that does not include any reflections. In such a case, transformations are simply given by $\varsigma_4$.
More generally, the state space of $m$-in $2$-out gbits is given by an $m$-dimensional hypercube and the transformations are $\varsigma_{2^{m-1}} \rtimes C_2$ where $\varsigma_{2^{m-1}}$ is the symmetric group of degree $2^{m-1}$.

Due to the lack of symmetry, a non-quantum state space generally has more restrictions on the allowed transformations than a qubit.
In Ref.~\cite{GrossMCD10}, it has been shown that all reversible transformations of gbits correspond to relabelling of the outcomes or the individual gbits, leading to the conclusion that the reversible computational power of gbits cannot exceed that of classical computational power.

\subsection{Spekkens' toy model} \label{Sec:STM}
In this section, we discuss a popular hidden variable model first introduced by Robert Spekkens, which has properties similar to quantum theory\cite{Spekkens07,vanEnk07}.

For a single Spekkens' bit, we consider a hidden variable which can take one of four possible values, which we label $1$, $2$, $3$ and $4$. We constrain our knowledge of the system such that we can only determine the state of the system down to one of one of two possible different states. We label a state where the hidden variable could be in $i$ or $j$ by the expression $i \lor j$. If we were to visualise the ontic system as being a ball in one of four possible slots arranged in a grid, we can only ascertain which row, or which column or which diagonal the ball is in (but never more than one of these facts simultaneously).

It is possible to, by analogy with a qubit, assign labels to these states of knowledge such that they represent the outcome of one of three possible binary measurements ($X$, $Y$ and $Z$). A full set of possible labels for the single bit case is shown in Table \ref{tab:spekkens}. These states are known as the {\em epistemic states} of the system, as they describe our knowledge of the system.

\begin{table}[h]
 \setlength{\unitlength}{0.5cm}
 \begin{center}
  \begin{tabularx}{6.5cm}{ | p{0.5cm}p{0.5cm}p{0.5cm} | X | c | }
\hline
\multicolumn{3}{|c|}{Measurement} & \centering{ Hidden} & Visual\\
\multicolumn{3}{|c|}{outcome} & \centering{variable} & representation\\
\hline
\raggedleft{X} &\centering=& +1	&  \centering{$1 \lor 3$} & \raisebox{-0.15cm}{\spek[13]} \\
\raggedleft{X} &\centering=& -1 &  \centering{$2 \lor 4$} & \raisebox{-0.15cm}{\spek[24]} \\
\raggedleft{Y} &\centering=& +1 	&  \centering{$1 \lor 4$} & \raisebox{-0.15cm}{\spek[14]} \\
\raggedleft{Y} &\centering=& -1 &  \centering{$2 \lor 3$} & \raisebox{-0.15cm}{\spek[23]} \\
\raggedleft{Z} &\centering=& +1 	&  \centering{$1 \lor 2$} & \raisebox{-0.15cm}{\spek[12]} \\
\raggedleft{Z} &\centering=& -1 &  \centering{$3 \lor 4$} & \raisebox{-0.15cm}{\spek[34]} \\
\hline
  \end{tabularx}
   \caption{\label{tab:spekkens} One possible labelling for single-bit states in Spekkens' toy theory.}
 \end{center}
\end{table}

The model has a well-defined measurement update rule, such that when a measurement is made on the state, we randomise the hidden variable, so that it is equally likely to be in either of the possible ontic states corresponding to the measurement. This means, for example, if we were to measure the state $Z=1$ ($1\lor2$) in the $X$ basis, we would (with equal probablity) get a value of $X=1$, putting our system into state $1\lor3$, or get a value of $X=-1$, putting our system into state $2 \lor 4$. In this sense, an uncertainty principle is built into the model- as we can never know the exact position of the hidden variable, we are restricted to knowing at most one of $X$, $Y$ or $Z$ simultaneously.

Allowed reversible transformations in this theory correspond to the permutation of the hidden variable. As the hidden variable can take four possible outcomes, this is the simplex group $S_4$. We note that if the operation $g \in S_4$ acts on the hidden variable $g: i \to g(i)$, then the operation on the epistemic states is to take $i \lor j \to g(i) \lor g(j)$. We note that because $g \in S_{4}$, each $i$ is taken to a unique value, and so if $i \neq j$, then $g(i) \neq g(j)$. Hence, acting on the underlying hidden variable with $S_4$, we map valid epistemic states to valid epistemic states.

\inlineheading{Representation in the GPT framework}
The hidden variable is a single classical measurement with four possible outcomes (a 4-simplex), and its normalised state-space in the GPT framework is the hull spanned by the ontic states $\left\{ (1, 0, 0, 0)^\mathrm{T}, (0, 1, 0, 0)^\mathrm{T}, (0, 0, 1, 0)^\mathrm{T}, (0, 0, 0, 1)^\mathrm{T} \right\}$. 
However, as these ontic states are never directly measurable this is not the most practical convex representation in which to consider the behaviour of Spekkens' toy model.
Instead, we recall that there are three allowed binary measurements for a single Spekkens' bit, and so we can plot the four ontic states in a 6-dimensional representation corresponding to the outcomes of these three binary measurements:  $\raisebox{-0.15cm}{\spek[1]}=\left(1, 0 | 1, 0 | 1, 0\right)^\mathrm{T}$, $\raisebox{-0.15cm}{\spek[2]}=\left(0, 1 | 0, 1 | 1, 0\right)^\mathrm{T}$, \mbox{$\raisebox{-0.15cm}{\spek[3]}=\left(1, 0 | 0, 1 | 0, 1\right)^\mathrm{T}$} and \mbox{$\raisebox{-0.15cm}{\spek[4]}=\left(0, 1| 1, 0 | 0, 1\right)^\mathrm{T}$}, as drawn in Figure \ref{fig:spek_tet_ontic}  in the expectation value picture. In this representation, each axis directly corresponds to the outcome of a different choice of measurement. These representations are isomorphic to each other, and so the set of homomorphisms on the ontic state space (and hence the set of allowed reversible transformations) is the same for both of them.

\begin{figure}[ht]
\centering
    \begin{subfigure}[t]{0.5\textwidth}
                \centering
                \includegraphics{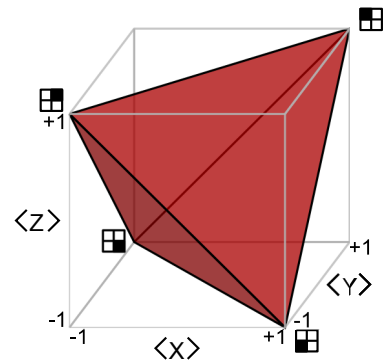}
                \caption{Representation of the ontic (hidden variable) states of Spekkens' toy model (tetrahedron).}      
	\label{fig:spek_tet_ontic}
	\end{subfigure}
	\begin{subfigure}[t]{0.5\textwidth}
	\centering
	\includegraphics{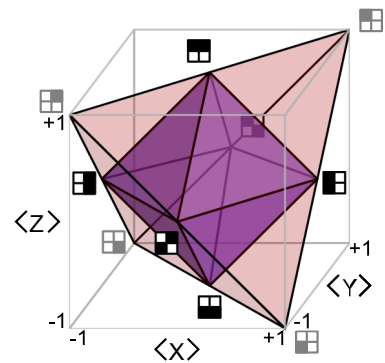}
	\caption{Representation of the epistemic states of Spekkens' toy model (octahedron).}
	\label{fig:spek_tet_epistemic}
\end{subfigure}
\caption{Spekkens' toy model state spaces.}
\end{figure}

In the convex framework, an equiprobable mixture of two states is represented as the half-way point on the line between the states. Thus, we can place the labels for states $i \lor j$ on the lines between points $i$ and $j$ and so label the epistemic states of the theory as in Figure \ref{fig:spek_tet_epistemic}. We note that in this representation, the epistemic states will have all the correct statistics for the model: each corner of the octahedron takes an extremal value for one measurement, and is totally mixed in all the others. Further more, if one were able to prepare a pure ontic state, and then perform any of the typical epistemic measurements, the state would project to the correct epistemic state in each case.

The allowed reversible transformations of the model must conform to the symmetry group of the hidden variable tetrahedron, rather than the embedded octahedral. From this visualisation, it is clear to see why some transformations are allowed and others are forbidden in the model: not all symmetries in the octahedron are also symmetries of the tetrahedron.  For example, consider the cyclic rotation of $X+ \to Y+ \to X- \to Y- \to X+$, whilst keeping $Z+$ and $Z-$ unchanged. This would be an allowed symmetry of the octahedron; but is actually forbidden in the Spekkens' model, as can be seen by the action this would have on the hidden variable: such a $90\degree$ rotation around the Z axis is not a symmetry of the tetrahedron, but would rather take the ontic state space onto the mirror image of itself.

By construction, the state space of an octahedron embedded within a tetrahedron shares the measurement statistics and transformation rules of Spekkens' toy model, and so is a valid representation in the convex framework.

%
%
\section{Phase group} \label{Sec:PG}

In this section, we generalize the concept of phase in the context of GPTs. We first consider phase in quantum theory in Sec.~\ref{SS:PGQ} and then generalize it into GPTs in Sec.~\ref{SS:PGGPT}.

\subsection{Phase group in quantum theory} \label{SS:PGQ}
For simplicity, we deal with a pure state $\ket{ \phi}$ in a Hilbert space $\mathcal{H}=\mathbb{C}^D$ with dimension $D$.
The state $\ket{\phi}$ expanded in a given basis $\Upsilon = \{ \ket{u_a} \}$ is given by $\ket{\phi}= \sum r_a e^{i \phi_a} \ket{u_a}$ where 
amplitudes $\{ r_a \}$ define the probabilities when the measurements are performed in the basis $\Upsilon$.
We define a set of unitaries $G_\Phi^{\Upsilon}$ that only change the arguments of the complex coefficients in the basis $\Upsilon$, that is, 
\begin{equation}
G_\Phi^{\Upsilon} = \left\{ \sum_{a=0}^{D-1} e^{i \phi_a} \ketbra{u_a}{u_a}  \right\} _{ \{ \phi_a \} \in [0,2 \pi )^{\times D} }. \label{Eq:quant}
\end{equation}
The set of unitary operations $G_\Phi^{\Upsilon}$ forms a group and 
we refer to the group as a {\it phase group associated with the basis $\Upsilon$}.

The elements of the phase group associated with the basis $\Upsilon$ do not change the probability distribution of measurement outcomes performed in the basis $\Upsilon$. Based on this interpretation, the phase group $G_\Phi^{\Upsilon}$ in quantum theory is understood as follows.
Let $\Upsilon=\{ \ket{u_a} \}$ be a basis in a Hilbert space $\mathcal{H}$.
The phase group associated with the basis $\Upsilon = \{ \ket{u_a} \}$, $G_\Phi^{\Upsilon}= \{ U_k^{\Upsilon} \}_k$, is the maximum subgroup of a unitary group of which all elements $U_k^{\Upsilon}$ satisfy that $\forall \ket{\phi} \in \mathcal{H}$ and $\forall a, k$,
\begin{equation}
|\braket{u_a}{\phi}|^2=|\bra{u_a} U_k^{\Upsilon}  \ket{\phi}|^2.
\end{equation}

More generally, we can consider phase groups associated with a positive operator valued measure (POVM) $\{ M_i \}$, such that $0\leq \mathrm{tr} M_i\rho \leq 1\,\text{and} \sum_i \mathrm{tr} M_i \rho=1 \ \forall \rho , M_i$, where $\rho$ is a density matrix with trace 1. 
There is a neat characterisation of the phase group with respect to these more general measurements: Let $\mathcal{H}$ be a Hilbert space and $\mathcal{B}(\mathcal{H})$ be a set of states. The phase-group associated with a POVM $\{ M_i \}$ is a subgroup of the unitary group acting on the Hilbert space and all elements should satisfy
\begin{equation}
\forall \rho \in \mathcal{B}(\mathcal{H}), \forall i, \hspace{1em} \mathrm{tr} M_i \rho = \mathrm{tr} M_i U \rho U^{\dagger},
\end{equation}
which implies
\begin{eqnarray}
\forall i, & \hspace{1em} & U^{\dagger} M_i U  = M_i, \\ 
\forall i, & \hspace{1em} & [M_i, U]  = 0.
\end{eqnarray}

Thus we see that the phase group associated with a POVM consists of the maximal set of unitary operators that commute with each element of the POVM.
If all elements in the POVM $\{ M_i \}$ are commutable ($\forall i,j,\ [M_i, M_j]=0$), the phase group is given by $\{ \sum_k e^{\i \theta_k} \ketbra{m_k}{m_k} \}_{\theta_k \in [0, 2\pi)}$ where $\{\ket{m_k}\}$ is the common eigenbasis of the POVM $\{M_i\}$, so it takes the same form as the phase group of a projective measurements. 

By considering phase groups with respect to all POVMs, one can extrapolate between the group of all unitaries and the identity. The phase group with respect to any informationally complete measurement is the identity. Consider for example choosing to measure in the eigenbasis of X, Y or Z at random when measuring a qubit (without forgetting which measurement was chosen). When one outcome is assigned to each of the six possibilities $X=\pm 1, Y=\pm 1, Z=\pm 1$, these probabilities uniquely specify the state such that no phase transforms are possible. Conversely, the phase group with respect to the single outcome $\{M\}=\id$ (`is there a state?')  is the full group of unitaries.

\subsection{Phase group in the GPT framework} \label{SS:PGGPT}
We generalize the idea of the phase group for any GPT. 

\begin{defin}[Phase group $G_\Phi$] \label{Def:PG}
Consider a given GPT with a state space and group of allowed reversible transformations, $G$. Let $\{\vec{e}_i\}_{i=1}^M$ be a measurement. The {\it phase group} $G_\Phi$ associated with the measurement is the maximal subgroup of all transformations $T\in G$ that leave all outcome probabilities of the measurement invariant, that is, for any states $\vec{s}$ and $\forall i$,
\begin{equation}
\vec{e_i} \cdot \vec{s}=\vec{e_i} \cdot (T \vec{s}).
\end{equation} 
\end{defin} 

The phase groups are characterized by the measurement $\{\vec{e}_i\}_{i=1}^M$. In particular, the number of effects in the measurement, $M$, determines the amount of freedom in the phase group. In general, larger $M$ imposes more restrictions on the group so that the phase group tends to be smaller.

A special type of phase group is one associated with a {\em maximal measurement}. There may be alternative definitions one may consider, but for concreteness and clarity (particularly in Section~\ref{Sec:nontrivial}) we refer only to the following definition:

\begin{defin}[Maximal measurement] \label{Def:MM}
A maximal measurement is one which distinguishes the maximal number of pure states possible for the theory in question.
\end{defin}

In the following sections, we show the phase groups in three example state spaces: classical bits, qubits and gbits. 

\inlineheading{Classical bits}
For a single classical bit, the phase group of a maximal measurement is composed only of an identity operator since there exists only one fiducial measurement. In this case, we say that the phase group is {\em trivial} as it only contains the identity operator.

For a set of $N > 1$ classical bits, it is possible to have a non-trivial phase group for non-maximal measurements. For example, if we measure the parity of the system, permutations of the bits and bit-flips made only in pairs will not change the parity, and so these operations are in the phase group of the parity measurement.

In an even more trivial example, if we have a system of two bits and make a measurement on the first bit, any operation on the second bit will not affect the first bit, and so such reversible transformations on the second bit are in the phase group of this non-maximal measurement. 
What sets apart quantum theory from the classical sceneraio is not that non-trivial operations are possible, but rather that they are possible even when the measurement is maximal.

\inlineheading{Qubits}
For qubits, the phase group $G_\Phi$ is a subgroup of $SO(3)$. 
We first consider a phase group associated with a maximal measurement.
Since the state space for qubits is isotropic, we just consider the phase group associated with the $Z$ measurement without loss of generality. 
A transformation that does not change a probability distribution of outcomes of the $Z$ measurement is a rotation on the $X$-$Y$ plane, namely, $SO(2)$ with an axis in the $Z$ direction. The corresponding unitary operations are $diag\{e^{i \phi_0},  e^{i \phi_1} \}$ in the $Z$ basis, so that it coincides with the phases in quantum theory.

\inlineheading{Gbits}
We consider the phase group in $m$-in $n$-out gbits and demonstrate that non-abelian phases appear in $m$-in $2$-out gbits for $m \geq 4$ (or $m \geq 3$, if reflections are included in the theory's group of allowed transformations).
We also show that the phase group depends on the choice of measurement since a state space for gbits is not isotropic.\\

\inlinesubheading{$3$-in $2$-out gbits} 
We first consider the phase group associated with a fiducial measurement $\{\vec{e}_0, \vec{e}_1 \}$:
\begin{align}
\label{eq:z0_effect} \vec{e}_0 &= (0,0 ~|~ 0,0 ~|~ 1,0)^T, \\
\label{eq:z1_effect} \vec{e}_1 &= (0,0 ~|~ 0,0 ~|~ 0,1)^T,
\end{align}
which corresponds to the $Z$ measurement. Then, the phase group does not vary $p(i|Z)$ $(i=0,1)$ and changes only $p(i|X)$ and $p(i|Y)$ $(i=0,1)$. 
Since $p(1|W)=1- p(0|W)$, the phase group is composed of transformations that mix $p(0|X)$ and $p(0|Y)$. 
Recalling that the phase group $G_\Phi$ is a subgroup of $\mathcal{T}_6  \cong  \varsigma_4 \rtimes C_2$, the phase group associated with the $Z$ measurement is given by a group of the symmetry of a square, which is a dihedral group of order $4$, $\mathcal{D}_4$.

On the other hand, when a maximal measurement is not the fiducial measurement and it is given by
\begin{align}
\vec{f}_0 &= \frac{1}{3} (1,0 ~|~ 1,0 ~|~ 1,0)^T, \\
\vec{f}_1 &= \frac{1}{3} (0,1 ~|~ 0,1 ~|~ 0,1)^T,
\end{align}
$\sum_W p(0|W)$ should be invariant under the action of the phase group.
Then, the phase group is given by rotations along the axis vector $(1,1,1)$ in the state space, which is a dihedral group of order $3$, $\mathcal{D}_3$ (or alternatively the group of a 3-simplex, $S_3$). 

We also consider the phase groups associated with non-maximal measurements composed of $M$ effects.
As a trivial example, when $M=1$, the measurement is given by $\{ \frac{1}{3} (1,1 ~|~ 1,1 ~|~ 1,1)^T \}$. Thus, the phase group is equivalent to all transformations, that is, $G_\Phi=\mathcal{T}_6$. When the measurement contains $M=4$ effects given by
\begin{multline}
\Bigg\{ \frac{1}{2}(1,0~|~0,0~|~0,0), \frac{1}{2}(0,1~|~0,0~|~0,0),\\ \frac{1}{2}(0,0~|~1,0~|~0,0), \frac{1}{2}(0,0~|~0,1~|~0,0) \Bigg\},
\end{multline}
it is straightforward to see that the phase group consists only of a reflection with respect to the $X$-$Y$ plane and the identity element since the phase group can vary the probability distribution of the outcomes of the $Z$ measurement.

Thus, the larger $M$ results in a smaller phase group. The $M$ determines the degree of freedom that should be invariant under the action of the phase group.
Note that, the phase groups are in general non-abelian for $m\geq3$ if we allow a reflection as a transformation. However, 
if we take an analogy with quantum theory and exclude reflections, then all phase groups in $3$-in $2$-out gbits are abelian. \\

\inlinesubheading{$4$-in $2$-out gbits} 
A state of a $4$-in $2$-out gbit is given by $\vec{s}=(p(i|W))_{\{i=0,1; W=X_0,X_1, X_2, X_3\}} $ and the state space is a $4$-dimensional hypercube. 
The transformations of the $4$-in $2$-out gbits are given by $\varsigma_{8} \rtimes C_2$:  the set of all maps from a $4$-dimensional hypercube to itself. We show that there exists a non-trivial non-abelian phase group for $4$-in $2$-out gbits. 

Firstly, we consider the phase group associated with maximal measurements.
When we take one of the fiducial measurements as a maximal measurement, the phase group is given by a polyhedral group of order $6$, $\mathcal{T}_6$.  To see this, consider a phase group associated with the $X_i$ measurements. The phase group changes
the probability distribution of other fiducial measurements. Then, the phase group in $4$-in $2$-out gbits associated with a maximal measurement $X_i$ ($i=0,1,2,3$) is equivalent to the group of all transformations of $3$-in $2$-out gbits, which is given by $\mathcal{T}_6$.
If we take a maximal measurement that is not one of the fiducial measurements, the phase group differs from $\mathcal{T}_6$. For instance, when the maximal measurement is given by $\{ \frac{1}{4}(1,0~|~1,0~|~1,0~|~1,0), \frac{1}{4}(0,1~|~0,1~|~0,1~|~0,1)\}$, the phase group is $\mathcal{S}_4$.

The phase groups $\mathcal{T}_6$ and $\mathcal{S}_4$ are non-abelian groups since they represent a symmetry of a cube and that of a tetrahedron, respectively. The symmetry group of $d$ dimensional objects is non-abelian when $d \geq 3$. 
Thus, $m$-in $2$-out gbits ($m >3$) have non-trivial non-abelian phase groups.

Nonetheless, there exist abelian phase groups in the $4$-in $2$-out gbits.
Let us consider the non-maximal measurement given by 
\begin{equation}
\{ \mathbf{g}_1,  \mathbf{g}_2, \mathbf{g}_3, \mathbf{g}_4\} = \frac{1}{2}
\Bigg\{
\begin{pmatrix}
0 \\
0 \\ \hline
0 \\
0 \\ \hline
0 \\
0 \\ \hline
1 \\
0 
\end{pmatrix},
\begin{pmatrix}
0 \\
0 \\ \hline
0 \\
0 \\ \hline
0 \\
0 \\ \hline
0 \\
1 
\end{pmatrix},
\begin{pmatrix}
0 \\
0 \\ \hline
0 \\
0 \\ \hline
1 \\
0 \\ \hline
0 \\
0 
\end{pmatrix},
\begin{pmatrix}
0 \\
0 \\ \hline
0 \\
0 \\ \hline
0 \\
1 \\ \hline
0 \\
0 
\end{pmatrix}\Bigg\}.
\end{equation}
In this case, the action of the phase group preserves the measurement outcomes of $X_2$ and $X_3$. 
The remaining state space is a square defined by two variables $p(0|X_0)$ and $p(0|X_1)$. 
Hence, the corresponding phase group is the rotational symmetry of a square, which is abelian if 
we exclude reflections.

More simply, we could also consider a set of six effects giving the probabilities associated with measurements $X_1$, $X_2$ and $X_4$, such that the phase group $C_2$ corresponds only to flipping the probabilities associated with $X_3$.

\inlineheading{Spekkens' toy model}
In Spekkens' toy theory, if we consider the phase group formed by fixing one measurement (such as the Z direction), we obtain the subgroup of permutations $(12)(34)$ which do not change the epistemic states associated with Z. This group, $Z_2 \oplus Z_2$, corresponds to either swapping the top two ontic states or swapping the bottom two ontic states. This phase group is in agreement with the group of allowed operations as described by Coecke et al.\cite{CoeckeES11}.

It might be possible to consider a more exotic measurement in the model, where we measure on the diagonal axis associated with the effects 
$\vec{s}_0 = (1,0~|~1,0~|~1,0)^T$ and $\vec{s}_1 = (0,1~|~0,1~|~0,1)^T$. The most extremal points in the octahedral phase-space are equiprobable mixtures of $\Big($\raisebox{-0.15cm}{\spek[12]}, \raisebox{-0.15cm}{\spek[24]} and \raisebox{-0.15cm}{\spek[14]}$\Big)$ and the equiprobable mixture of $\Big($\raisebox{-0.15cm}{\spek[34]},~\raisebox{-0.15cm}{\spek[13]}~and~\raisebox{-0.15cm}{\spek[23]}$\Big)$. In both state spaces, the whole system has three-fold rotational symmetry about this axis, plus three planes of reflective symmetry: exactly the symmetries of a triangle. Thus, we see that such a measurement has the phase group $S_3$.

We remark that in the tetrahedral space there is an asymmetry, as the second of these states is also an extremal point in the tetrahedral space, whereas the tetrahedron extends beyond the first to include the corner state $(1,0~|~1,0~|~1,0)^T$- this means that although we may be able to measure to distinguish between these two states, there does not exist a valid linear operation in the framework which can exchange them. 

One possible interpretation of such a measurement and its associated phase group would be to say that the extremal points of our measurement are caused by a three-way mixture of measurements, which we could perform by choosing uniformally randomly which of the three primary bases (X, Y or Z) to measure (assuming that making any of these measurements will collapse the system to a ``pure'' epistemic state on the octahedron) and then taking our result, but discarding any information about which basis we used. From this process, it's clear that we have the freedom to permutate the labellings of the bases without affecting our result (which makes it almost self-evident that the phase group should be the permutation group $S_3$), so long as we make sure for each constitutent measurement, we're only comparing the outcomes made in the same basis.

\section{Irreversible phase dynamics}
As well as reversible dynamics, which lead to the natural group structure as discussed in Section~\ref{Sec:PG}, it is possible to consider other operations that may be inherently non-reversible, but still preserve the evaluated output with the a set of effects of some measurement. We refer to this sort of operation as being part of the {\em phase dynamics}- where the term ``phase'' is drawn by analogy with the measurement-preserving nature of the operation.
As these operations have no unique inverse they do not form a group structure, but rather form a semi-group, much like the set of completely-positive maps acting on a density matrix in quantum mechanics.

\subsection{Examples}
\inlineheading{Quantum decoherence}
In quantum theory, this is analogous to {\em decoherence}. Consider the state corresponding to a pure X eigenstate:
\begin{equation}
\mathbf{s} = \begin{pmatrix}
1 \\
0 \\ \hline
1/2 \\
1/2 \\ \hline
1/2 \\
1/2\end{pmatrix},
\end{equation}

Consider the operation $D$ that replaces all $X$ and $Y$ statistics with $(1/2, 1/2)$ (such as, for example, flipping the $X$ state with probability $1/2$, or leaving it unchanged, with probability $1/2$) will change the state into the maximally mixed state:
\begin{equation}
\mathbf{s_\mathrm{mix}} = \begin{pmatrix}
1/2 \\
1/2 \\ \hline
1/2 \\
1/2 \\ \hline
1/2 \\
1/2\end{pmatrix},
\end{equation}

Considering the effects of the $Z$ measurement $\vec{e}_0$, $\vec{e}_1$, (defined in Eqns.~\ref{eq:z0_effect} and \ref{eq:z1_effect}), we see that $\vec{e}_i \cdot \mathbf{s} = \vec{e}_i \cdot \mathbf{s_\mathrm{mix}}$- and so the statistics associated with such a measurement is unchanged. In quantum theory, this would correspond with replacing a coherent superposition in some basis with a classical mixture displaying the same measurement statistics for one basis.

One can also consider applying different elements of the phase group with some classical probability. As none of the phase group operations disturb the measurement associated with the phase group, the composite operation will also preserve this measurement. For example, if one combines the operations of a small unitary $Z$ rotation around the Bloch sphere with some small random chance of making a jump across to the other side (i.e. a 180\,\degree~Z rotation) the joint transformation corresponds to a path inwardly spiralling around the Bloch sphere, preserving the Z statistics.

\inlineheading{`Measurement setting' on a gbit}
A related operation that is mathematically possible on a gbit (but not realisable on a qubit), is to always set the $X$ statistics of a system to $(1,0)$ without changing the statistics of any other measurements (in some ways making the state `more pure'), such as by the operation $P$,
\begin{equation}
P = \begin{pmatrix}
1 & 1 & 0 & 0 & 0 & 0 \\
0 & 0 & 0 & 0 & 0 & 0 \\
0 & 0 & 1 & 0 & 0 & 0 \\
0 & 0 & 0 & 1 & 0 & 0 \\
0 & 0 & 0 & 0 & 1 & 0 \\
0 & 0 & 0 & 0 & 0 & 1
\end{pmatrix}
\end{equation}

\subsection{Phase dynamics are non-trivial only for non-classical state spaces}
\label{Sec:nontrivial}
Apart from considering examples of theories and whether they have non-trivial phase dynamics one may hope to make a more general statement concerning which features of a theory endow it with non-trivial phase. 
 
We define as is standard a theory as classical if a state is uniquely specified by the statistics for a single measurement with which it is possible to distinguish $N$ pure states (using the standard definition of $N$).  

We define the phase dynamics associated with a measurement as the set of all dynamics which leave the statistics of the measurement in question invariant. (Note that this may include irreversible dynamics).

\begin{thmnonumb}
Phase dynamics associated with a maximal measurement are non-trivial iff the theory is non-classical.
\end{thmnonumb}

\begin{proof}We break it up into two cases.\\
{\em (i) theory is classical:} In this case a maximal measurement having 
its statistics frozen means the full state is frozen, thus the only allowed phase dynamics 
is the identity $\mathbbm{1}$, which changes no state. Thus if a theory is classical only trivial such dynamics exist.

{\em (ii) theory is non-classical:} We need to show that non-trivial phase dynamics always exists in this case. 
We take without loss of generality one of the fiducial measurements of the state vector to be maximal (implying it has $N$ outcomes). 
We take, without loss of generality, the first maximal fiducial measurement to be the one frozen. 
As we know $K>N$ there are still some free parameters associated with one or more additional fiducial measurements.

Consider the following transform: take the first maximal measurement. Take N effects: $\bra{e_1}$,$\bra{e_2}...\bra{e_N}$ (note that the bra-ket notation is now used for {\em real} vectors). Take N states  $\ket{\mu_1}...\ket{\mu_N}$ such that $\braket{e_i}{\mu_j}=\delta_{ij}\,\forall i,j$. Let the transform be
$$ T=\sum_i \ket{\mu_i}\bra{e_i}.$$
We want to show:
\begin{enumerate}
\item T always exists in a non-classical theory and is allowed (making the implicit assumption that any dynamics taking states to states are allowed). 
\item T constitutes phase dynamics of the maximal measurement: it leaves the measurement statistics invariant.
\item T is always non-trivial: it changes at least one state.
\end{enumerate}    

1. To prove that $T$ always exists we note the following. We can take
\begin{eqnarray}
 \bra{e_1}&=&(1\, 0...0\,|\,0...0) \\
 &\vdots & \\
 \bra{e_N}&=&(0...0\, 1\,|\,0...0).
\end{eqnarray}

These effects always exist and always yield probabilities summing to one for any states. We can moreover take 
\begin{eqnarray}
 \ket{\mu_1}&=&(1\, 0...0\,|\text{anything allowed})^T \\
 &\vdots & \\
 \ket{\mu_N}&=&(0...0\, 1\,|\text{anything allowed})^T.
\end{eqnarray}
These are always allowed states as we have assumed the state space contains N maximally distinguishable states associated with the first fiducial measurement.
We see that $\braket{e_i}{\mu_j}=\delta_{ij}\,\forall i,j$. T is an allowed transform as it is a matrix and takes states to states: it takes any state to a mixture of the $\ket{\mu_i}$ states, which is allowed as the states are allowed and all mixtures of allowed states are allowed. 

2. T is an example of phase dynamics associated with the first fiducial measurement by the following argument. Consider an arbitrary state $\ket{\eta}$. Then the probability of any outcome of the frozen measurement is given by
$$\braket{e_i}{\eta}.$$
After the transform we have 
\begin{eqnarray}
\bra{e_i}T\ket{\eta}&=&\bra{e_i}\sum_j \ket{\mu_j}\braket{e_j}{\eta}\\
&=&\sum_j \delta_{ij}\braket{e_j}{\eta}\\
&=& \braket{e_i}{\eta},
\end{eqnarray}
As T preserves the statistics of the measurement it is in the set of associated phase dynamics. 

3. T is always non-trivial. If there is a classical system there are some free parameters apart from those defined by 
the statistics of the maximal first fiducial measurement. For some particular distribution of the first measurement at least 
two possible states exist, call them $\ket{\eta_1}$ and $\ket{\eta_2}$. Yet T will output the same state for both of those input states:
\begin{equation}
T\ket{\eta_1}=T\ket{\eta_2},
\end{equation}
which is because the final state is uniquely determined by the probabilities of the maximal measurement for the input state. 
Thus T must change at least one of the states $\ket{\eta_1}$ and $\ket{\eta_2}$.
\end{proof}

\section{Interference} \label{Sec:Int}
In  this section, we show that in quantum theory the phase group plays an important role in systems that are said to exhibit interference. 
Thus we formulate quantum interference in the GPT framework, and extend this process to be applicable to all GPTs.

\subsection{Quantum interference}

\inlineheading{Young's double slit experiment}
In a single-photon version of the Young's double-slit experiment, the output measurement is no longer a binary variable, but instead encapsulates a continuous range of possible positions where the photon could land on the screen. The common physical meaning of the term `interference' describes the pattern which forms on this screen, which can not be determined just by considering the sum of spatial distribution probabilities from each slit in turn.

Adding a piece of glass in front of one of the slits changes the overall pattern, without changing the output distributions seen if each slit is considered on its own. Some part of the set-up has been changed without disturbing the output distribution statistics of each slit- the addition of glass to change which interference pattern we observe is therefore a {\em phase operation}.

\inlineheading{Mach-Zehnder interferometer}
A simpler example of a device exhibiting interference is the Mach-Zehnder interferometer (MZI) as illustrated in the single qubit circuit presented in Fig.~\ref{Fig:MZ}. 
Consider an initial state prepared in the computational basis: $\ket{0}$. Through unitary operations, the initial state is transformed to:
\begin{equation}
\ket{f} = U_H^{\dagger} \begin{pmatrix} 1 & 0 \\ 0 & e^{i\phi} \end{pmatrix} U_H \ket{0}, \label{Eq:Time}
\end{equation}
where $U_H=U_H^\dagger=\frac{1}{\sqrt{2}} \begin{pmatrix} 1 & 1 \\ 1 & -1 \end{pmatrix}$ is the Hadamard gate. 

\begin{figure}[ht]
\centering
  \includegraphics[width=0.9\linewidth, clip]{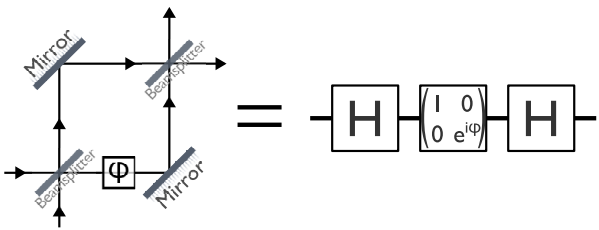}
  \caption{A Mach-Zehnder interferometer and \\ its equivalent quantum circuit.}
\label{Fig:MZ}
\end{figure}

We perform a measurement on the final state $\ket{f}$ in the computational basis $\{\ket{0}, \ket{1} \}$. The probability to obtain outcome $Z=\pm1$ is $P(Z\-=\-\pm1)=\frac{1 \pm \cos \phi}{2}$, which gives different probabilities as a function of the phase shift $\phi$. 
Just as in Young's double slits, the Z measurement statistics for the output of the MZI may be considered to be an {\em interference pattern}. Had the beam-splitters simply mixed the photon in a decoherent manner, then the output distribution would not depend on $\phi$.

\subsection{Quantum interference in the GPT framework}
The MZI circuit can be described in the GPT framework as follows:

The initial state $\ket{0}$ is represented by:
\begin{align}
\vec{s}_0 = \left(
\begin{array}{c}
P(+1|X) \\ 
P(-1|X) \\
\hline
\vspace{-0.25cm} \\
P(+1|Y) \\ 
P(-1|Y) \\
\hline
\vspace{-0.25cm} \\
P(+1|Z) \\
P(-1|Z)
\end{array} \right)
= \left( \begin{array}{c}
1/2 \\ 
1/2 \\
\hline
\vspace{-0.25cm} \\
1/2 \\ 
1/2 \\
\hline
\vspace{-0.25cm} \\
1 \\
0
\end{array} \right)
\end{align}
The unitaries $U_H$ and $\begin{pmatrix} 1 & 0 \\ 0 & e^{i\phi} \end{pmatrix}$ can be described by their actions on the probability vector:
\begin{eqnarray}
U_H & \leftrightarrow T_H:= \begin{pmatrix} 
0 & 0 & 0 & 0 & 1 & 0 \\
0 & 0 & 0 & 0 & 0 & 1 \\
0 & 0 & 0 & 1 & 0 & 0 \\ 
0 & 0 & 1 & 0 & 0 & 0 \\
1 & 0 & 0 & 0 & 0 & 0 \\ 
0 & 1 & 0 & 0 & 0 & 0 
\end{pmatrix}, \label{eq:GPT_H}
\end{eqnarray}

\begin{eqnarray}
\begin{pmatrix} 
1 & 0 \\
 0 & e^{i\phi} 
\end{pmatrix} & \leftrightarrow & T_{\phi} :=
\end{eqnarray}

\begin{equation}
\tiny
\begin{pmatrix} 
\dfrac{\lambda_1 +\cos \phi}{2} & \dfrac{\lambda_1-\cos \phi}{2} & \dfrac{(1\!-\!\lambda_1) -\sin \phi}{2} & \dfrac{(1\!-\!\lambda_1)+\sin \phi}{2} \\ 
\dfrac{\lambda_2 -\cos \phi}{2} & \dfrac{\lambda_2+\cos \phi}{2} & \dfrac{(1\!-\!\lambda_2) +\sin \phi}{2} & \dfrac{(1\!-\!\lambda_2)-\sin \phi}{2} \\
\dfrac{\lambda_3 +\sin\phi}{2} & \dfrac{\lambda_3-\sin\phi}{2} & \dfrac{(1\!-\!\lambda_3) +\cos \phi}{2} & \dfrac{(1\!-\!\lambda_3) - \cos \phi}{2} \\ 
\dfrac{\lambda_4 -\sin\phi}{2} & \dfrac{\lambda_4+\sin\phi}{2} & \dfrac{(1\!-\!\lambda_4) -\cos \phi}{2} & \dfrac{(1\!-\!\lambda_4)+\cos \phi}{2}
\end{pmatrix} \oplus \id_2, \nonumber
\end{equation}
which is a representation of SO(2) on the space of probability vectors, where $\lambda_1$, $\lambda_2$, $\lambda_3$, $\lambda_4$ can be any numerical value, as when $T$ acts on a normalised probablity vector, these unphysical degrees of freedom always disappear. A clear (but still arbitrary) choice would be $\lambda_1=\lambda_2=1$, $\lambda_3=\lambda_4=0$, such that when $\phi=0$, the matrix has the form of the diagonal identity matrix; but it should be noted that even for other choices of $\lambda$, the matrix will have no effect on the probability vectors when $\phi=0$.

The final state $\vec{s}_f$ is therefore given by:
\begin{align}
\vec{s}_f &=T_H T_{\phi} T_H \vec{s}_0 \label{Eq:Int} \\
&= \left( \begin{array}{c}
1/2 \\ 
1/2 \\
\vspace{-0.25cm} \\
\hline
\vspace{-0.25cm} \\
\frac{1}{2}(1-\sin \phi) \\
\frac{1}{2}(1+\sin \phi) \\
\vspace{-0.25cm} \\
\hline
\vspace{-0.25cm} \\
\frac{1}{2}(1+\cos \phi) \\
\frac{1}{2}(1-\cos \phi)
\end{array} \right).
\end{align}
Performing the $Z$ measurement, we obtain the outcome $+1$ with a probability $\frac{1}{2}(1+\cos \phi)$ and the outcome $-1$
with a probability $\frac{1}{2}(1-\cos \phi)$.

\subsection{Interference in other GPTs}
By analogy with interference in quantum theory, we define interference in general based on Eq.~\eqref{Eq:Int}. We assume that for one of the measurements in the theory, $Z$, we are allowed to directly prepare states in and measure (e.g.\ position). We then also require the existence of at least one `beamsplitter'-like transformation $T_H$ (and its inverse $T_H^{-1}$, which might be equal to $T_H$) which relates the statistics of $Z$ with some of the other statistics of the state (and vice versa). The simplest case is to swap the $Z$ measurement statistics with the statistics of some other measurement. Finally, we need a set of transformations $\{ T_\Phi \}$, which is the phase group associated with $Z$.

\begin{defin}[Interference in GPTs] \label{Def:Int}
For a measurement $\mathcal{E}$ with associated phase group $G_\Phi^{\mathcal{E}}$, we can construct a compound transformation on an initial state $\vec{s}_0$:
\begin{equation}
\vec{s}_{f} = (T_H^{-1}\circ g_e  \circ T_H)\vec{s}_{0}, \notag
\end{equation}
where $g_e \in G_\Phi^{\mathcal{E}}$, and $T_H$ is defined as above (e.g.~a Hadamard gate). 

If the statistics of $\mathcal{E}$ in state $\vec{s}_f$ depend on the choice of phase group element $g_e$, then we say that the theory demonstrates {\em non-trivial interference}, and the statistics of $\mathcal{E}$ in $\vec{s}_f$ are the {\em interference pattern} associated with the choice of $g_e$.
\end{defin}

We see that the phase group is naturally related (by conjugation with $T_H$) to the set of allowed interference patterns.

\inlineheading{$3$-in $2$-out gbits}
We show that by this definition, 3-in 2-out gbits can exhibit interference. 
Let $G_\Phi^{\mathcal{Z}}$ be the phase group associated with the $Z$ measurement and $T_H$ be (identically to the quantum Hadamard):

\begin{equation}
T_H = T_H^{-1} =\begin{pmatrix}
0 & 0 & 0 & 0 & 1 & 0 \\
0 & 0 & 0 & 0 & 0 & 1 \\
0 & 0 & 0 & 1 & 0 & 0 \\
0 & 0 & 1 & 0 & 0 & 0 \\
1 & 0 & 0 & 0 & 0 & 0 \\
0 & 1 & 0 & 0 & 0 & 0
\end{pmatrix}.
\end{equation}

For initial state $\vec{s}_0$, we consider the evolution
\begin{equation}
\vec{s}_{f} = (T_H^{-1}\circ g^Z_i \circ T_H)(\vec{s}_{0}), 
\end{equation}
where $g^Z_i \in G_\Phi^{\mathcal{Z}}$- the phase group of Z (the 8 automorphisms of a square, consisting of 90\degree\,rotations around the Z-axis and reflections in the planes XZ and YZ).

We explicitly label the elements of $G_\Phi^{Z} = \{g^Z_i | i = 1\ldots8 \}$ where
\begin{eqnarray}
g^Z_1=\begin{pmatrix}
0 & 0 & 0 & 1 \\
0 & 0 & 1 & 0 \\
1 & 0 & 0 & 0 \\
0 & 1 & 0 & 0 
\end{pmatrix} \oplus \id_2 \\
g^Z_k=(g^Z_1)^k | k= 1 \ldots 4. \\
g^Z_5=\begin{pmatrix}
0 & 1 & 0 & 0 \\
1 & 0 & 0 & 0 \\
0 & 0 & 1 & 0 \\
0 & 0 & 0 & 1 
\end{pmatrix} \oplus \id_2  \\ 
g^Z_k= (g^Z_1)^{k-5}(g^Z_5)^k | k= 5 \ldots 8. \\
\end{eqnarray}
We can think of $g^Z_1$, $g^Z_2$, $g^Z_3$, $g^Z_4$ as rotations, and $g^Z_5$,~$g^Z_6$,~$g^Z_7$,~$g^Z_8$ as a flip followed by a rotation. The full set of transformations, and the final states of $H g^Z_i H \vec{s}_0$ are listed in Appendix~\ref{App:SquareGroup}.

\begin{table}[ht]
\caption{The probability distributions to obtain $Z=+1$~or~$-1$ for the different phase group elements $g^Z_i$\\ in a 3-in~2-out gbit.}
\begin{tabular}{c|c|c}
$g^Z_i$ & Prob. to obtain $+1$ & Prob. to obtain $-1$ \\ \hline
$g^Z_1$  & $p(+1|Y)$ & $p(-1|Y)$\\
$g^Z_2$  & $p(-1|Z)$ & $p(+1|Z)$\\
$g^Z_3$  & $p(-1|Y)$ & $p(+1|Y)$\\
$g^Z_4$  & $p(+1|Z)$ & $p(-1|Z)$\\
$g^Z_5$  & $p(-1|Z)$ & $p(+1|Z)$\\
$g^Z_6$  & $p(+1|Y)$ & $p(-1|Y)$\\
$g^Z_7$  & $p(+1|Z)$ & $p(-1|Z)$\\
$g^Z_8$  & $p(-1|Y)$ & $p(+1|Y)$
\end{tabular} \label{Tab:Interference}
\end{table}

Finally, by performing the $Z$ measurement on our output state, we obtain the probability distribution presented in Table.~\ref{Tab:Interference}. For some input states, the final measurement outcomes will depend on our choice of $g_e$, so this procedure has the ability to display different {\em interference patterns}. 
We note that the output statistics do not distinguish between the application of phase group members $g^Z_1$ or $g^Z_6$,  $g^Z_2$ or $g^Z_5$, $g^Z_3$ or $g^Z_8$, and $g^Z_4$ or $g^Z_7$. If we know in advance what phase group member we have chosen, such interferometry can be used to tell us about the statistics of the $Y$ or $Z$ measurements in the initial state. To determine the statistics of the $X$ measurement, we would have to pick a different $T_H$.

\inlineheading{Spekkens' toy model}
In Spekken's model, it is possible to choose a $T^{\mathrm{spek}}_H$ which has some of the same behaviour as quantum the Hadamard gate. If we want the gate to be self inverse, and map $X=\pm1$ to $Z=\pm1$ and back again, the best we can do is a permutation $1324$ swapping around the second and third ontic states. Unlike the quantum Hadamard, this transformation will not change $Y$ states. Acting on probability vectors, this gate is expressed:

\begin{equation}
T^{\mathrm{spek}}_H = (T^{\mathrm{spek}}_H)^{-1} =\begin{pmatrix}
0 & 0 & 0 & 0 & 1 & 0 \\
0 & 0 & 0 & 0 & 0 & 1 \\
0 & 0 & 1 & 0 & 0 & 0 \\
0 & 0 & 0 & 1 & 0 & 0 \\
1 & 0 & 0 & 0 & 0 & 0 \\
0 & 1 & 0 & 0 & 0 & 0
\end{pmatrix}.
\end{equation}

The phase group of permutations preserving $Z$ measurements is the set of four permutations $(12)(34)$, which is a form of $Z_2\oplus Z_2$, represented as transformations on probabilities as:

\begin{eqnarray}
g_{1234} & = & \id_6 \\
g_{2134} & = & \begin{pmatrix}
0 & 0 & 0 & 1 \\
0 & 0 & 1 & 0 \\
0 & 1 & 0 & 0 \\
1 & 0 & 0 & 0
\end{pmatrix} \oplus \id_2 \\
g_{1243} & = & \begin{pmatrix}
0 & 0 & 1 & 0 \\
0 & 0 & 0 & 1 \\
1 & 0 & 0 & 0 \\
0 & 1 & 0 & 0
\end{pmatrix} \oplus \id_2 \\
g_{2143} & = & \begin{pmatrix}
0 & 1 & 0 & 0 \\
1 & 0 & 0 & 0 \\
0 & 0 & 0 & 1 \\
0 & 0 & 1 & 0
\end{pmatrix} \oplus \id_2 
\end{eqnarray}

Thus, by considering the effect of  $(T^{\mathrm{spek}}_H)^{-1} \cdot g_e \cdot T^{\mathrm{spek}}_H$ on a generic input state, we obtain a probability distribution for outputs as listed in Table~\ref{Tab:IX_Spek}.

\begin{table}[ht]
\caption{The probability distributions to obtain $Z=+1$~or~$-1$ for the different phase group elements $g_e$\\ in Spekkens' toy model.}
\begin{tabular}{c|c|c}
$g_e$ & Prob. to obtain $+1$ & Prob. to obtain $-1$ \\ \hline
$g_{1234}$  & $p(+1|Z)$ & $p(-1|Z)$\\
$g_{2134}$  & $p(-1|Y)$ & $p(+1|Y)$\\
$g_{1243}$  & $p(+1|Y)$ & $p(-1|Y)$\\
$g_{2143}$  & $p(-1|Z)$ & $p(+1|Z)$
\end{tabular} \label{Tab:IX_Spek}
\end{table}

However, it should be noted in Spekkens' toy model (as it is for qubits), if we prepare an initial state to have a well defined outcome in one of $Z$ or one of $Y$ outcomes, then the other measurement will be uniformly random- and so from the outcomes in Table~\ref{Tab:IX_Spek}, at best we will only be able to tell three possibilities of $g_e$ apart, even after performing repeated tests.

\subsection{Interference in branching interferometers}
There is one class of interferometer we call {\em branching interferometers}, in which a particle is directed down one of many possible paths, disjoint in space.  The MZI is a branching interferometer with two such paths, but this can be generalised to higher number of `branches' of the interferometer. 
The particle travelling through the system could be directed down paths spatially a long distance away from each other. It is natural by reasons of non-signalling to forbid local operations that cause the particle to jump from one disjoint branch to another, and thus the set of allowed operations after splitting must be in the phase group of the `which branch' measurement.

In such a system, it is tempting to consider a set of operations that act on one of the branches in a local manner (such as adding a piece of glass on one branch). For example, in quantum theory on a three-branch system, one could execute an operation $U_\mathrm{upper} = \mathrm{diag}\left(e^{i\phi_1}, 1, 1\right)$ on the upper branch,  $U_\mathrm{middle} = \mathrm{diag}\left(1, e^{i\phi_1}, 1\right)$ on the middle branch,  or $U_\mathrm{lower} = \mathrm{diag}\left(1, 1, e^{i\phi_3}\right)$ on the lower one, and all three of these elements will contribute towards the total phase group. It can be shown that no operation performed on the middle branch will ever adjust the relative phase between upper and lower branches, for example, and so we say that some set of operations are {\em localised} to a sub-region of the system.

However, for theories with non-abelian phase group elements, it is dangerous to trivially consider specific members of the phase group as local operations. Consider a non-abelian phase group $\mathcal{G}_Z$, with two elements $a, b \in \mathcal{G}_Z$ such that $[a,b]\neq0$. If we say $a$ applies at some point on the upper branch and $b$ applies at some point along a disjoint lower branch, we note that because $a\cdot b \neq b \cdot a$, the order in which these operations are applied will, for some states, affect the final statistics when the branches are brought back together.

\begin{figure}[ht]
\centering
  \includegraphics[width=0.9\linewidth, clip]{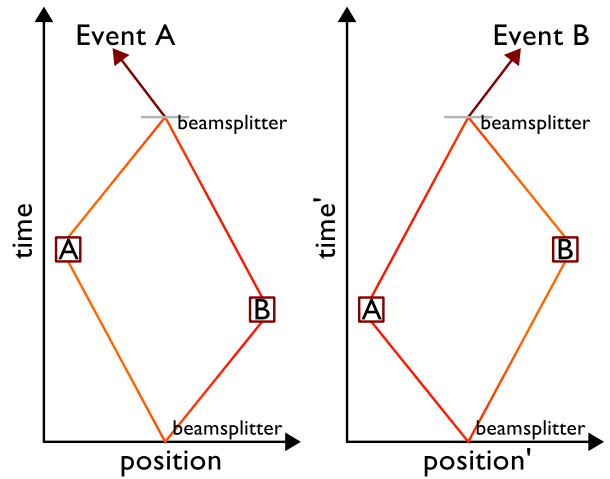}
  \caption{A two-branch interferometer shown in two frames of reference. A and B are space-like separated, and so the order in which these they are applied depends on the observer. If A and B do not commute, then different final events will be seen by the different observers.}
\label{Fig:RoS}
\end{figure}

As illustrated in Fig~\ref{Fig:RoS}, if the measured statistic is the $Z$ measurement of the particle's position after a beam-splitter, triggering a cascade in a particle detector, it is reasonable to assume these statistics should be Lorentz invariant. For two spatially disjoint operations on the branches, because of {\em relativity of simultaneity} there will be frames of reference in which the operations occur in different orders, and thus predict different output statistics. 
This violates the assumption of a single objective reality (that is, will potentially affect parameters which should be Lorentz invariant), and thus unmodified non-commuting elements of the phase group can not be considered as local operations happening on different branches. 

This does not rule out non-abelian phase group elements within a branch locally. Consider $\{a_1, \ldots a_n, b_1 \ldots b_n\} \in G_z$ where $[a_i,a_j]\neq0$ and $[b_i,b_j]\neq0$ but $[a_i,b_j] = 0$. Here, it would be possible to assign the sub-group of operations $\{a\}$ to be local to one branch, and $\{b\}$ to be local on another without running into simultaneity problems, because all the non-commuting elements are time-like separated and so have a well established ordering.

It is possible to add other physical conditions on local actions in branching interferometers. This results in more restrictions being placed on the choice of $g_e \in G^Z_\phi$, and is discussed in depth in ref.~\cite{DahlstenGV12}.

\section{Summary and concluding remarks} \label{Sec:Sum}
We defined phase operationally using the GPT framework. This allowed us to investigate phase in theories other than quantum theory. We found that phase is ubiquitous, in the sense that any non-classical theory has non-trivial phase transforms (where phase is defined with respect to a so-called maximal measurement). We determined the groups of reversible phase transforms for examples of theories other than quantum theory, finding for example that some theories have non-abelian phase groups (with respect to maximal measurements), unlike quantum theory. We discussed how phase relates to interference in GPTs.

The aim of this work was to lay the foundations for studying phase in GPTs. We now anticipate that these definitions and methods will be used to investigate connection between phase and other phenomena, such as computational speed-up and thermalisation. 

\begin{acknowledgments}
The authors thank Bob Coecke, Aleks Kissinger, Markus M\"uller, Shojun Nakayama, Felix Pollock, Matt Pusey, Takanori Sugiyama and Eyuri Wakakuwa for useful discussions. 

This work was supported by Project for Developing Innovation Systems of the Ministry of Education, Culture, Sports, Science and Technology (MEXT), Japan. 
Y.~N. acknowledges support from JSPS by KAKENHI (Grant No. 222812) 
and M.~M. acknowledges support from JSPS by KAKENHI  (Grant No. 23540463 and No. 23240001).  
O.~D., A.~G. and V.~V. acknowledge funding from the National Research Foundation (Singapore), the Ministry of Education (Singapore), the EPSRC (UK), the Templeton Foundation, the Leverhulme Trust, and the Oxford Martin School.

\end{acknowledgments}

\newpage
\bibliography{phase_refs}
\bibliographystyle{h-physrev}

\newpage $\,$
\appendix

\section{Derivation of the form of effects associated with a general pure quantum state}
\label{app:general_effect}
Consider the basis states $\{ \ket{e},\ket{e^{\perp}} \}$ where $\ket{e} = (\cos \alpha, e^{i\beta} \sin \alpha)^\mathrm{T}$ and $\ket{e^{\perp}} = (\sin \alpha, -e^{i\beta} \cos \alpha)^\mathrm{T}$. We wish to derive the effects associated with these states.

One method is to consider rotating the states $\{\ket{e}, \ket{e^\perp}\}$ into the  $\{\ket{0}, \ket{1}\}$ basis, with some transformation $T$, and then making a Z measurement using the effects associated with the $Z$  measurement.

To do this, first we construct a unitary operator $T$ in Hilbert space that transforms the general states to the computational basis such that $T\ket{e} = \ket{0}$, and $T\ket{e^\perp} = \ket{1}$:
\begin{eqnarray}
T & = & \ketbra{0}{e} + \ketbra{1}{e^\perp} \nonumber \\
& = & \left( \begin{array}{cc}
\cos\left( \frac{\alpha}{2} \right) & \sin\left( \frac{\alpha}{2} \right) \\
e^{i\beta} \sin \left( \frac{\alpha}{2} \right) & -e^{i\beta}  \cos\left( \frac{\alpha}{2} \right)
\end{array} \right)
\end{eqnarray}
There is more than one transformation that will operationally switch the states $\ketbra{0}{0}$ with $\ketbra{e}{e}$, etc. but we have arbitrarily chosen the state that does not add an additional phase term to simplify the mathematics.

For a general pure state: $\ket{\psi} = \cos\left(\frac{\zeta}{2}\right) \ket{0} + \sin\left(\frac{\zeta}{2}\right) e^{i\phi} \ket{1}$, we see $\ket{\psi} \to \ket{\psi'} = T\ket{\psi}$ is given:
\begin{eqnarray*}
\ket{\psi'}  & = & \left( \begin{array}{cc}
\cos\left( \frac{\alpha}{2} \right) & \sin\left( \frac{\alpha}{2} \right) \\
e^{i\beta} \sin \left( \frac{\alpha}{2} \right) & -e^{i\beta}  \cos\left( \frac{\alpha}{2} \right)
\end{array} \right) \left( \begin{array}{c}
\cos\left(\frac{\zeta}{2}\right) \\
e^{i\phi} \sin\left(\frac{\zeta}{2}\right)
\end{array} \right) \nonumber \\
& = & \left( \begin{array}{c}
\cos\left(\frac{\zeta}{2}\right) \cos\left( \frac{\alpha}{2} \right) + \sin\left(\frac{\zeta}{2}\right) \sin\left( \frac{\alpha}{2} \right)   e^{i\phi} \\
\cos\left(\frac{\zeta}{2}\right) \sin \left( \frac{\alpha}{2} \right) e^{i\beta} - \sin\left(\frac{\zeta}{2}\right)  \cos\left( \frac{\alpha}{2} \right) e^{i\beta}    e^{i\phi}
\end{array} \right)
\end{eqnarray*}
We want to consider the operational effect of $T$ (i.e.~how it changes a given set of measurement outcome probabilities). We use the expectation value picture to simplify the calculation and find $T_\mathrm{expt}$:

\begin{widetext}
\begin{eqnarray*}
\vec{\nu}_{\psi'} & = & T_\mathrm{expt}\left( \begin{array}{c}
\expt{X} \\
\expt{Y} \\
\expt{Z}\end{array} \right) \\
& = &  
\left(  \begin{array}{c}
- \cos\left(\alpha\right)  \cos\left(\beta\right) \cos\left(\phi\right) \sin\left(\zeta\right)
- \cos\left(\alpha\right)  \sin\left(\beta\right) \sin\left(\phi\right) \sin\left(\zeta\right)
+\sin\left( \alpha \right)  \cos\left( \zeta \right) \\

\sin\left( \beta \right) \cos\left( \phi \right)  \sin\left(\zeta\right) 
- \cos\left( \beta \right) \sin\left( \phi \right)  \sin\left(\zeta\right) \\

\sin\left(\alpha\right)  \cos\left(\beta\right) \cos\left(\phi\right) \sin\left(\zeta\right)
+ \sin\left(\alpha\right)  \sin\left(\beta\right) \sin\left(\phi\right) \sin\left(\zeta\right)
+ \cos\left( \alpha \right)  \cos\left( \zeta \right) \\

\end{array} \right) \\
& = &
\left( \begin{array}{ccc}
-\cos\left(\alpha\right) \cos\left(\beta\right) & -\cos\left(\alpha\right) \sin\left(\beta\right) & \sin\left(\alpha\right) \\
\sin\left(\beta\right) & -\cos\left(\beta\right) & 0 \\
\sin\left(\alpha\right) \cos\left(\beta\right) & \sin\left(\alpha\right) \sin\left(\beta\right) & \cos\left(\alpha\right)
\end{array} \right) 
\left( \begin{array}{c}
\sin\left(\zeta\right) \cos\left(\phi\right) \\
\sin\left(\zeta\right) \sin\left(\phi\right) \\
\cos\left(\zeta\right)
\end{array} \right).
\end{eqnarray*}
\end{widetext}

Thus we see that, as expected from quantum theory, the expectation value matrix has been acted on by an element of SO(3).

We rewrite this element in terms of the action upon the probabilities. We note that $\expt{X} = 2P(X=1)-1$, and embed the transformation into the bigger matrix $1\oplus T_\mathrm{expt}$ such that it acts on the vector $(1, \expt{X}, \expt{Y}, \expt{Z})$, where the first element $1$ is a normalisation term. Thus, we can convert the transformation on expectations to one acting on probabilities and vice versa sing the transform $C$, and its inverse $C^{-1}$:

\begin{eqnarray*}
C = \left( \begin{array}{cccccc}
A & A & B & B & C & C \\
1 & -1 & 0 & 0 & 0 & 0 \\
0 & 0 & 1 & -1 & 0 & 0 \\
0 & 0 & 0 & 0 & 1 & -1
\end{array} \right), \\
C^{-1} = \left( \begin{array}{cccc}
\frac{1}{2} & \frac{1}{2} & 0 & 0 \\
\frac{1}{2} & -\frac{1}{2} & 0 & 0 \\
\frac{1}{2} & 0 & \frac{1}{2} & 0 \\
\frac{1}{2} & 0 & -\frac{1}{2} & 0 \\
\frac{1}{2} & 0 & 0 & \frac{1}{2} \\
\frac{1}{2} & 0 & 0  & -\frac{1}{2} \\
\end{array} \right).
\end{eqnarray*}
We leave in the unphysical excess parameters $A$,$B$,$C$ where $A+B+C=1$, which arise from an extra degree of freedom in our transformation resultant from a restriction on the state vectors to form a set of probabilities.

It is hence possible to construct the general action on the probability vector $T_\mathrm{prob}$, as shown in (Equation~\ref{eq:gentrans} in Appendix~\ref{app:general_unitary}).

Thinking of an effect as a co-vector associated with a state, we write the effects associated with the $Z$ measurement as $\vec{e_{+}} = \left(0, 0 ~|~ 0, 0 ~|~ 1, 0\right)$ and $\vec{e_{-}} = \left(0, 0 ~|~ 0, 0 ~|~ 0, 1\right)$, and post-multiply them by $T_\mathrm{prob}$.

\begin{eqnarray}
\label{eq:gen_effect_1}
\vec{e} = \left( 
\begin{array}{c}
\frac{1}{2} (A+\cos\left(\beta\right) \sin\left(\alpha\right))  \\
\frac{1}{2} (A-\cos\left(\beta\right) \sin\left(\alpha\right))  \\
\frac{1}{2} (B+\sin\left(\alpha\right) \sin\left(\beta\right))  \\	
\frac{1}{2} (B-\sin\left(\alpha\right) \sin\left(\beta\right))  \\
\frac{1}{2} (C+\cos\left(\alpha\right)) \\
\frac{1}{2} (C-\cos\left(\alpha\right))
\end{array}
\right) \\
\label{eq:gen_effect_2}
\vec{e_\perp} = \left( 
\begin{array}{c}
\frac{1}{2} (A-\cos\left(\beta\right) \sin\left(\alpha\right))  \\
\frac{1}{2} (A+\cos\left(\beta\right) \sin\left(\alpha\right))  \\
\frac{1}{2} (B-\sin\left(\alpha\right) \sin\left(\beta\right))  \\	
\frac{1}{2} (B+\sin\left(\alpha\right) \sin\left(\beta\right))  \\
\frac{1}{2} (C-\cos\left(\alpha\right)) \\
\frac{1}{2} (C+\cos\left(\alpha\right))
\end{array} \right)
\end{eqnarray}

\newpage
\section{General unitary transformation on a qubit probability space}
\label{app:general_unitary}
In general, the unitary transformation acting on a qubit to map $\ket{0}\to\ket{e}=(\cos \alpha, e^{i\beta} \sin \alpha)^\mathrm{T}$ and $\ket{1}\to\ket{e_\perp} =(\sin \alpha, -e^{i\beta} \cos \alpha)^\mathrm{T}$ is given by the following matrix, (the derivation of which is outlined in Appendix\ref{app:general_effect})):
\begin{widetext}

\begin{eqnarray}
\tiny
\label{eq:gentrans}
T_\mathrm{prob} = \left(
\begin{array}{cccccc}
 \frac{1}{2} (A-\cos \left(\alpha\right) \cos \left(\beta\right)) & \frac{1}{2} (A+\cos \left(\alpha\right) \cos \left(\beta\right)) & \frac{1}{2} (B-\cos \left(\alpha\right) \sin \left(\beta\right)) & \frac{1}{2} (B+\cos \left(\alpha\right) \sin \left(\beta\right)) & \frac{1}{2} (C+\sin \left(\alpha\right)) & \frac{1}{2} (C-\sin \left(\alpha\right)) \\
 \frac{1}{2} (A+\cos \left(\alpha\right) \cos \left(\beta\right)) & \frac{1}{2} (A-\cos \left(\alpha\right) \cos \left(\beta\right)) & \frac{1}{2} (B+\cos \left(\alpha\right) \sin \left(\beta\right)) & \frac{1}{2} (B-\cos \left(\alpha\right) \sin \left(\beta\right)) & \frac{1}{2} (C-\sin \left(\alpha\right)) & \frac{1}{2} (C+\sin \left(\alpha\right)) \\
 \frac{1}{2} (A+\sin \left(\beta\right)) & \frac{1}{2} (A-\sin \left(\beta\right)) & \frac{1}{2} (B-\cos \left(\beta\right)) & \frac{1}{2} (B+\cos \left(\beta\right)) & \frac{C}{2} & \frac{C}{2} \\
 \frac{1}{2} (A-\sin \left(\beta\right)) & \frac{1}{2} (A+\sin \left(\beta\right)) & \frac{1}{2} (B+\cos \left(\beta\right)) & \frac{1}{2} (B-\cos \left(\beta\right)) & \frac{C}{2} & \frac{C}{2} \\
 \frac{1}{2} (A+\cos \left(\beta\right) \sin \left(\alpha\right)) & \frac{1}{2} (A-\cos \left(\beta\right) \sin \left(\alpha\right)) & \frac{1}{2} (B+\sin \left(\alpha\right) \sin \left(\beta\right)) & \frac{1}{2} (B-\sin \left(\alpha\right) \sin \left(\beta\right)) & \frac{1}{2} (C+\cos \left(\alpha\right)) & \frac{1}{2} (C-\cos \left(\alpha\right)) \\
 \frac{1}{2} (A-\cos \left(\beta\right) \sin \left(\alpha\right)) & \frac{1}{2} (A+\cos \left(\beta\right) \sin \left(\alpha\right)) & \frac{1}{2} (B-\sin \left(\alpha\right) \sin \left(\beta\right)) & \frac{1}{2} (B+\sin \left(\alpha\right) \sin \left(\beta\right)) & \frac{1}{2} (C-\cos \left(\alpha\right)) & \frac{1}{2} (C+\cos \left(\alpha\right))
\end{array}
\right)
\end{eqnarray}
\end{widetext}

\newpage

\section{Explicit phase group of Z in a 3-in 2-out gbit}
\label{App:SquareGroup}
The elements of the phase group $G^Z_\Phi$ associated with the Z measurement can be written explicitly:
\begin{eqnarray}
g^Z_1=\begin{pmatrix}
0 & 0 & 0 & 1 \\
0 & 0 & 1 & 0 \\
1 & 0 & 0 & 0 \\
0 & 1 & 0 & 0 
\end{pmatrix} \oplus \id_2 \\
g^Z_2=\begin{pmatrix}
0 & 1 & 0 & 0 \\
1 & 0 & 0 & 0 \\
0 & 0 & 0 & 1 \\
0 & 0 & 1 & 0 
\end{pmatrix} \oplus \id_2 \\
g^Z_3=\begin{pmatrix}
0 & 0 & 1 & 0 \\
0 & 0 & 0 & 1 \\
0 & 1 & 0 & 0 \\
1 & 0 & 0 & 0 
\end{pmatrix} \oplus \id_2 \\
g^Z_4=\begin{pmatrix}
1 & 0 & 0 & 0 \\
0 & 1 & 0 & 0 \\
0 & 0 & 1 & 0 \\
0 & 0 & 0 & 1 
\end{pmatrix} \oplus \id_2 \\
g^Z_5=\begin{pmatrix}
0 & 1 & 0 & 0 \\
1 & 0 & 0 & 0 \\
0 & 0 & 1 & 0 \\
0 & 0 & 0 & 1 
\end{pmatrix} \oplus \id_2 \\
g^Z_6=\begin{pmatrix}
0 & 0 & 0 & 1 \\
0 & 0 & 1 & 0 \\
0 & 1 & 0 & 0 \\
1 & 0 & 0 & 0 
\end{pmatrix} \oplus \id_2 \\
g^Z_7=\begin{pmatrix}
1 & 0 & 0 & 0 \\
0 & 1 & 0 & 0 \\
0 & 0 & 0 & 1 \\
0 & 0 & 1 & 0 
\end{pmatrix} \oplus \id_2 \\
g^Z_8=\begin{pmatrix}
0 & 0 & 1 & 0 \\
0 & 0 & 0 & 1 \\
1 & 0 & 0 & 0 \\
0 & 1 & 0 & 0 
\end{pmatrix} \oplus \id_2
\end{eqnarray}
These matrices should not be mistaken for the unitary operators acting on a Hilbert space- they are transformations operating on the probability vectors.

For a state initially in $\vec{s_0}$, where
\begin{equation}
\vec{s_0} = \left( \begin{array}{c}
P(+1|X) \\
P(-1|X) \\ \hline
P(+1|Y) \\
P(-1|Y) \\ \hline
P(+1|Z) \\
P(-1|Z) 
\end{array} \right),
\end{equation}

we see these transformation has the following effect on the statistics:

\begin{equation}
T_H^{-1} g^Z_1 T_H \vec{s_0}  = \left( \begin{array}{c}
P(+1|X) \\
P(-1|X) \\ \hline
P(-1|Z) \\
P(+1|Z) \\ \hline
P(+1|Y) \\
P(-1|Y) \\
\end{array} \right)
\end{equation}
\begin{equation}
T_H^{-1} g^Z_2 T_H  \vec{s_0}  =  \left( \begin{array}{c}
P(+1|X) \\
P(-1|X) \\ \hline
P(-1|Y) \\
P(+1|Y) \\ \hline
P(-1|Z) \\
P(+1|Z) 
\end{array} \right)
\end{equation}
\begin{equation}
T_H^{-1} g^Z_3 T_H  \vec{s_0}  =  \left( \begin{array}{c}
P(+1|X) \\
P(-1|X) \\ \hline
P(+1|Z) \\
P(-1|Z) \\ \hline
P(-1|Y) \\
P(+1|Y) 
\end{array} \right) 
\end{equation}
\begin{equation}
T_H^{-1} g^Z_4 T_H  \vec{s_0}  =  \left( \begin{array}{c}
P(+1|X) \\
P(-1|X) \\ \hline
P(+1|Y) \\
P(-1|Y) \\ \hline
P(+1|Z) \\
P(-1|Z) 
\end{array} \right) 
\end{equation}
\begin{equation}
T_H^{-1} g^Z_5 T_H  \vec{s_0}  =  \left( \begin{array}{c}
P(+1|X) \\
P(-1|X) \\ \hline
P(+1|Y) \\
P(-1|Y) \\ \hline
P(-1|Z) \\
P(+1|Z) 
\end{array} \right)
\end{equation}
\begin{equation}
T_H^{-1} g^Z_6 T_H  \vec{s_0}  =  \left( \begin{array}{c}
P(+1|X) \\
P(-1|X) \\ \hline
P(+1|Z) \\
P(-1|Z) \\ \hline
P(+1|Y) \\
P(-1|Y) 
\end{array} \right) 
\end{equation}
\begin{equation}
T_H^{-1} g^Z_7 T_H \vec{s_0}  =  \left( \begin{array}{c}
P(+1|X) \\
P(-1|X) \\ \hline
P(-1|Y) \\
P(+1|Y) \\ \hline
P(-1|Z) \\
P(+1|Z) 
\end{array} \right) 
\end{equation}
\begin{equation}
T_H^{-1} g^Z_8 T_H  \vec{s_0}  =  \left( \begin{array}{c}
P(+1|X) \\
P(-1|X) \\ \hline
P(-1|Z) \\
P(+1|Z) \\ \hline
P(-1|Y) \\
P(+1|Y) 
\end{array} \right) 
\end{equation}

\end{document}